\documentclass[11pt]{article}

\usepackage{fullpage}

\usepackage{amssymb}

\usepackage[ruled,linesnumbered,vlined]{algorithm2e}

\usepackage{graphicx}
\usepackage{amsthm}
\usepackage{amsmath}
\usepackage{amsfonts,mathtools}
\usepackage{booktabs}
\usepackage{paralist}
\usepackage{subfigure}
\usepackage{color}

\newtheorem{theorem}{Theorem}[section]

\newtheorem{definition}{Definition}[section]

\newtheorem{proposition}{Proposition}[section]
\newtheorem{corollary}{Corollary}[section]

\definecolor{auditred}{rgb}{0.60,0.10,0.10}

\newcommand{\PROP}{\mathsf{PROP}}
\newcommand{\WPROP}{\mathsf{WPROP}}

\newcommand{\bX}{\mathbf{X}}
\newcommand{\bY}{\mathbf{Y}}
\newcommand{\bZ}{\mathbf{Z}}

\begin{document}




\title{On the Incompatibility of Weighted PROPX and Pareto Optimality for Indivisible Chores}
\author{}
\date{}


%

\author{Haris Aziz$^1$, Bo Li$^2$\medskip \\
$^1$UNSW Sydney,
\texttt{haris.aziz@unsw.edu.au}\\
$^2$The Hong Kong Polytechnic University,
\texttt{comp-bo.li@polyu.edu.hk}}












\maketitle

\begin{abstract}

Proportionality (PROP) is one of the simplest fairness criteria for allocating items among agents with additive preferences. With indivisible chores, however, PROP is not always satisfiable. We study proportionality up to any item (PROPX), which requires every agent to satisfy proportionality after any chore is removed from her bundle.
Under strictly positive costs, we settle the weighted compatibility question negatively: weighted PROPX and Pareto optimality are incompatible already for two agents and four chores. Moreover, for every $n\geq3$, we give an $n$-agent, $(n+1)$-chore counterexample whose shares can be arbitrarily close to equal.
These counterexamples are item-minimal: under strictly positive costs, weighted PROPX and Pareto optimality are always compatible when the number of chores is at most the number of agents, and they are compatible for two agents with at most three chores.
Our impossibility result contrasts with the compatibility theorem of  Mahara (2026) for weighted envy-freeness up to one item (EF1) and Pareto optimality .



\bigskip


\end{abstract}

\section{Introduction}

Proportionality (PROP), introduced by Steinhaus for the division of a nonatomic cake \cite{steihaus1948problem}, is one of the most basic fairness requirements: each of $n$ agents should receive at least a $1/n$ share of the total value. The same principle applies naturally to the division of undesirable items, or \emph{chores}, where each agent should incur at most a $1/n$ share of her total cost.
For indivisible items, exact proportionality may not be satisfiable. This has motivated several relaxations in the rapidly growing literature on fair allocation without monetary transfers~\cite{Moul19a,DBLP:journals/ai/AmanatidisABFLMVW23}. One popular benchmark is the maximin-share (MMS) guarantee~\cite{conf/bqgt/Budish10}. Although each agent's MMS value is well defined, an allocation satisfying every agent's MMS guarantee need not exist~\cite{journals/jacm/KurokawaPW18,conf/aaai/AzizRSW17,DBLP:conf/wine/FeigeST21}. Moreover, computing an MMS allocation, when one exists, is NP-hard~\cite{conf/aaai/AzizRSW17}. Approximation algorithms for both goods and chores have therefore received substantial attention~\cite{DBLP:journals/corr/abs-2307-07304,akrami2023improving,DBLP:conf/sigecom/HuangS23}.

A simpler relaxation is proportionality up to one item (PROP1)~\cite{conf/sigecom/ConitzerF017}. For goods, an agent may add one suitable good outside her bundle; for chores, she may remove one suitable chore from it. PROP1 is compatible with Pareto optimality for both goods and chores~\cite{DBLP:conf/sigecom/BarmanKV18,journals/orl/AzizMS20}. It is nevertheless a relatively weak guarantee: envy-freeness up to one item (EF1) implies PROP1~\cite{conf/sigecom/LiptonMMS04}. For goods, EF1 is compatible with Pareto optimality~\cite{journals/teco/CaragiannisKMPS19,DBLP:conf/sigecom/BarmanKV18}. Mahara recently proved the corresponding compatibility for additive chores, including weighted EF1~\cite{DBLP:conf/soda/Mahara26}.

We study the stronger requirement of proportionality up to any item (PROPX). For chores, an allocation is PROPX if every agent becomes proportional after the removal of \emph{any} chore from her bundle. Unlike in the goods setting, where PROPX allocations may not exist~\cite{journals/orl/AzizMS20}, PROPX allocations of chores always exist and can be computed efficiently~\cite{Moul19a,DBLP:conf/www/0037L022}.
Existence alone, however, does not guarantee that the allocation is economically meaningful. A fair allocation may be inefficient by assigning chores to agents who incur unnecessarily high costs. Pareto optimality (PO) rules out this issue: an allocation is Pareto optimal if no other allocation weakly reduces every agent's cost and strictly reduces the cost of at least one agent. Compatibility with PO is therefore a fundamental problem of whether a fairness guarantee can be achieved without avoidable inefficiency.
{Aziz et al.~\cite{DBLP:journals/ai/AzizLMWZ24} showed that PROPX and PO can be incompatible when some agent--chore costs are zero. They left open the following central question:}

\begin{quote}
\textit{Are weighted PROPX and PO compatible when every chore has strictly positive cost to every agent?}
\end{quote}

{A key difficulty in answering the above question is that PROPX is incompatible with fractional Pareto optimality (fPO) even under strictly positive costs, as proved in \cite{DBLP:journals/ai/AzizLMWZ24}. Consequently, approaches that obtain efficiency by maintaining fPO, including market-equilibrium and related techniques used in recent work~\cite{DBLP:journals/corr/abs-2212-02440,DBLP:conf/soda/Mahara26}, cannot be applied directly; a solution must exploit integral Pareto optimality without imposing fractional efficiency.}

In this paper, we answer this question negatively. With unequal shares, weighted PROPX and PO are incompatible already for two agents and four chores. Since every cost in our examples is positive, the impossibility also holds for weak weighted PROPX, which imposes the condition only after removing a positive-cost chore. This contrasts with Mahara's compatibility theorem for weighted EF1 and Pareto optimality~\cite{DBLP:conf/soda/Mahara26}.

{For every $n\geq3$, we then construct an $n$-agent counterexample with $n+1$ chores in which the agents' shares are arbitrarily close to equal. Together with the two-agent example, this gives an incompatible weighted instance for every $n\geq2$. Under strictly positive costs, the constructions are item-minimal: compatibility always holds when the number of chores does not exceed the number of agents, and for two agents it also holds with three chores. Thus the failure is not caused by large instances or extreme asymmetry. To the best of our knowledge, the case of exactly equal shares and strictly positive costs remains open for $n\geq3$ when $m>n$.}

Unequal shares arise naturally when agents bear different levels of responsibility. Examples include employees with different job descriptions, university staff with different teaching obligations, and countries with different responsibilities for reducing carbon emissions. This weighted model is well established in cake cutting
\cite{robertson1998cake} and has more recently been studied for relaxations of envy-freeness
\cite{conf/atal/ChakrabortyISZ20,DBLP:conf/sigecom/0001Z023}
and maximin-share fairness
\cite{journals/jair/FarhadiGHLPSSY19,conf/ijcai/0001C019}.

\subsection{Main Results}
\label{sec:main-results}

{Our first result resolves the weighted efficiency question negatively under strictly positive additive costs. 
We construct a two-agent instance with four chores in which the unique weighted PROPX allocation is Pareto dominated.
We then extend the impossibility to every $n\geq3$. 
The counterexamples are also sharp with respect to the number of chores. 
Under strictly positive costs, a weighted PROPX and PO allocation always exists whenever $m\leq n$, for arbitrary positive shares. Hence the $(n+1)$-chore construction is item-minimal for every $n\geq3$. For two agents, we prove 
compatibility whenever $m\leq3$, and thus the two-agent four-chore construction is also item-minimal. These results separate the unequal-share setting from the equal-share setting: for two agents with equal shares, a PROPX and PO allocation always exists~\cite{DBLP:journals/ai/AzizLMWZ24}.
Our counterexamples do not settle the latter, and, to the best of our knowledge, the existence of a PROPX and PO allocation for equal shares and general strictly positive additive costs remains open for $n\geq3$ when $m>n$.}


\subsection{Other Related Work}
\label{sec:relatedworks}

{Aziz et al.~\cite{DBLP:journals/ai/AzizLMWZ24} showed that PROPX and PO may be incompatible when some agent--chore costs are zero. They further showed that, even when every agent--chore cost is strictly positive, PROPX may be incompatible with fractional Pareto optimality, which rules out Pareto improvements by fractional allocations. Since fPO is stronger than PO, this result does not resolve PROPX--PO compatibility. They also proved that, under strictly positive costs, PROPX and PO are compatible for two agents with equal shares and for any number of agents with lexicographic or bi-valued costs, even under unequal shares.}


\paragraph{Weighted Fairness}
Most of the fair-division literature assumes equal entitlement or obligation shares, while a growing literature considers arbitrary and possibly unequal shares.
For example, Farhadi et al.~\cite{journals/jair/FarhadiGHLPSSY19} and Aziz et al.~\cite{conf/ijcai/0001C019} adapted MMS to this setting for goods and chores, respectively, and designed approximation algorithms accordingly.
Babaioff et al.~\cite{journals/corr/abs-2103-04304} introduced AnyPrice Share (APS) for arbitrary entitlements and gave a $3/5$-approximation for goods.
Weighted EF1 allocations are known to exist and can be computed efficiently for goods \cite{conf/atal/ChakrabortyISZ20} and for chores \cite{DBLP:conf/sigecom/0001Z023}.
For chores, Aziz et al.~\cite{DBLP:journals/ai/AzizLMWZ24} showed that weighted PROPX implies a $2$-approximation to APS. Feige and Huang~\cite{DBLP:journals/corr/abs-2211-13951} subsequently improved the approximation factor to $1.733$.

\paragraph{Fairness vs. Efficiency}
{A central question is whether fairness and efficiency can be attained simultaneously. For goods, several compatibility results are known~\cite{DBLP:conf/sigecom/BarmanKV18,DBLP:journals/corr/abs-2204-14229,DBLP:conf/aaai/BarmanK19}.
For chores, earlier work established EF1 and PO for bi-valued costs~\cite{DBLP:conf/atal/EbadianP022,DBLP:conf/aaai/GargMQ22}. Garg et al.~\cite{DBLP:journals/corr/abs-2212-02440} proved the existence of EF1 and fPO allocations for three agents and for instances with at most two distinct cost functions. Mahara~\cite{DBLP:conf/soda/Mahara26} subsequently proved that EF1 and fPO are compatible for arbitrary additive chore costs and extended the result to weighted EF1.}


\section{Model and Solution Concepts}
\label{sec:preliminaries}

We consider the problem of fairly allocating a set of $m$ indivisible chores $M$ to a group of $n$ agents $N$.
Each agent $i\in N$ has a nonnegative additive cost function $c_i:2^M\to\mathbb{R}_{\geq0}$. Thus, $c_i(S)=\sum_{e\in S}c_i(e)$ for every $S\subseteq M$, where $c_i(e)$ abbreviates $c_i(\{e\})$.
Whenever $c_i(M)>0$ for every agent, we may normalize the costs so that $c_i(M)=1$.
An allocation is represented by a partition of the items $\bX=(X_1,\ldots,X_n)$, where each agent $i$ obtains $X_i$, $X_i\cap X_j = \emptyset$ for all $i\neq j$ and $\cup_{i\in N}X_i = M$.

\medskip

We next define {\em proportionality} and its relaxations. We first state the standard equal-share definitions and then extend proportionality to unequal shares.

\begin{definition}[PROP]
An allocation $\bX$ is proportional (PROP) if $c_i(X_i) \le \PROP_i$ for every $i\in N$, where $\PROP_i = c_i(M)/n$ is agent $i$'s proportionality.
\end{definition}

For normalized cost functions $\PROP_i = 1/n$ for all $i\in N$.

	




\begin{definition}[PROP1 and PROPX]
	\label{def:propx}
	An allocation $\bX$ is proportional up to one item (PROP1) if, for every $i\in N$ with $X_i\neq\emptyset$, there exists $e\in X_i$ such that $c_i(X_i\setminus\{e\})\leq\PROP_i$.
	The allocation is proportional up to any item (PROPX) if for any $i\in N$ and any $e\in X_i$, $c_i(X_i \setminus \{e\}) \le\PROP_i$.
\end{definition}
In the weighted setting, every agent $i$ has a share $s_i>0$, where $\sum_{i\in N}s_i=1$.

\begin{definition}[Weighted PROP1 and PROPX]
\label{def:weighted-propx}
Let $\WPROP_i=s_i c_i(M)$ denote agent $i$'s weighted proportionality.
An allocation $\bX$ is weighted PROP1 if, for every $i\in N$ with $X_i\neq\emptyset$, there exists $e\in X_i$ such that
$c_i(X_i\setminus\{e\})\leq \WPROP_i$.
It is weighted PROPX if, for every $i\in N$ and every $e\in X_i$,
$c_i(X_i\setminus\{e\})\leq \WPROP_i$.
When $s_i=1/n$ for every $i$, these definitions reduce to PROP1 and PROPX.
\end{definition}

We use the zero-tolerant version of weighted PROPX, which imposes the inequality above even when $c_i(e)=0$. Weak weighted PROPX imposes it only for chores $e\in X_i$ with $c_i(e)>0$. 
The two definitions coincide when $c_i(e)>0$ for every agent $i$ and chore $e$. Since our negative results hold for the case of strictly positive costs, it implies that our negative results also hold for the weaker version of PROPX.

Every PROPX allocation is PROP1, and every weighted PROPX allocation is weighted PROP1; the converses may not hold. For additive costs, PROPX and weighted PROPX allocations exist and can be found in polynomial time~\cite{Moul19a,DBLP:conf/www/0037L022}. We therefore focus on weighted PROPX allocations.

\begin{definition}[PO]
	An allocation $\bY=(Y_1,\ldots,Y_n)$ \emph{Pareto dominates} an allocation
	$\bX=(X_1,\ldots,X_n)$ if $c_i(Y_i)\leq c_i(X_i)$ for every $i\in N$, and $c_j(Y_j)<c_j(X_j)$ for at least one $j\in N$.
	An allocation $\bX$ is \emph{Pareto optimal} (PO) if it is not Pareto
	dominated by any allocation.
\end{definition}

%


\section{Weighted PROPX and Pareto optimality}
\label{sec:weighted-propx-po}

Throughout this section, all chore costs are strictly positive. For a nonempty bundle $S$, denote its PROPX residual for agent $i$ by $\rho_i(S)=c_i(S)-\min_{e\in S}c_i(e)$, and let $\rho_i(\emptyset)=0$. By additivity, an allocation $\bX$ is weighted PROPX if and only if $\rho_i(X_i)\leq s_i c_i(M)$ for every $i\in N$.

\subsection{Two agents}

\begin{theorem}
\label{thm:two-agent-weighted-counterexample}
Weighted PROPX and PO are incompatible even for two agents and four chores with strictly positive additive costs.
\end{theorem}

\begin{proof}
Let $N=\{1,2\}$ and $M=\{a,b,c,d\}$. Let the shares be $(s_1,s_2)=(3/10,7/10)$ and the costs be as follows.
\[
\begin{array}{c|cccc|c|c}
\text{Agent} &a&b&c&d&c_i(M)&s_i c_i(M)\\ \hline
1&12&1&10&10&33&99/10\\
2&11&6&1&6&24&84/5
\end{array}
\]
All costs are strictly positive.

Each of $a,c,d$ costs agent~$1$ more than $99/10$. Hence $|X_1|\leq1$ in every weighted PROPX allocation: any larger bundle contains one of these chores and another chore, whose removal leaves cost above agent~$1$'s share. The five possible bundles $X_1$ give the following residual costs for agent~$2$.
\[
\begin{array}{c|c|c}
X_1&X_2&\rho_2(X_2)\\ \hline
\emptyset&\{a,b,c,d\}&23\\
\{a\}&\{b,c,d\}&12\\
\{b\}&\{a,c,d\}&17\\
\{c\}&\{a,b,d\}&17\\
\{d\}&\{a,b,c\}&17
\end{array}
\]
Since $s_2c_2(M)=84/5$, the unique weighted PROPX allocation is $\bX=(\{a\},\{b,c,d\})$, with cost vector $(12,13)$. Let $\bY=(\{b,d\},\{a,c\})$. Its cost vector is $(11,12)$, so both agents strictly improve and $\bY$ Pareto dominates $\bX$. Thus no weighted PROPX allocation is PO.
\end{proof}

This is a sharp contrast to the case of two agents with equal shares, where a PROPX and PO allocation always exists~\cite{DBLP:journals/ai/AzizLMWZ24}. 

\subsection{More than two agents}

The following construction gives counterexamples for every $n\geq3$.

\begin{theorem}
\label{thm:share-cone}
Let $n\geq3$ and let $(s_1,\ldots,s_n)$ be a positive share vector. Suppose, after relabeling, that $s_i<1/n$ for every $i<n$, while $1/n<s_n<2/(n+2)$. Then there is an instance with $n$ agents, $n+1$ chores, and strictly positive normalized additive costs that admits no allocation that is both weighted PROPX and PO.
\end{theorem}

\begin{proof}
Let $M=\{a,p_1,\ldots,p_n\}$. Let agents $1,\ldots,n-1$ be the outer agents and agent $n$ the central agent. For each outer agent $i<n$, let
\[
0<\varepsilon_i<
\min\left\{\frac{1-ns_i}{2},\frac1{n+2}\right\},
\qquad
t_i=\frac{1-2\varepsilon_i}{n}.
\]
For the central agent, let $s_n<A<2/(n+2)$ and $r=(1-A)/n$. These parameters exist by the assumptions of the theorem.
The cost functions are given in the following table.
\[
\begin{array}{c|ccccc|cc}
\text{Agent} & a & p_1 & p_2 & \cdots & p_n & c_i(M) & s_i c_i(M)\\ \hline
1 & t_1+\varepsilon_1 & \varepsilon_1 & t_1 & \cdots & t_1 & 1 & s_1\\
2 & t_2+\varepsilon_2 & t_2 & \varepsilon_2 & \cdots & t_2 & 1 & s_2\\
\vdots & \vdots & \vdots & \vdots & \ddots & \vdots & \vdots & \vdots\\
n & A & r & r & \cdots & r & 1 & s_n
\end{array}
\]
That is, for all $i<n$, $c_i(a)=t_i+\varepsilon_i$, $c_i(p_i)=\varepsilon_i$ and $c_i(p_j)=t_i$ for all $j\neq i$.
For the central agent $n$, $c_n(a)=A$ and $c_n(p_j)=r$ for $j=1,\ldots,n$.
It can be verified that the row sum is $(t_i+\varepsilon_i)+\varepsilon_i+(n-1)t_i=nt_i+2\varepsilon_i=1$ for $i<n$ and $A+nr=1$ for the central agent. All entries are strictly positive. Hence each agent's weighted proportionality equals her share.

Consider a weighted PROPX allocation. For every outer agent $i<n$, $t_i-s_i=(1-ns_i-2\varepsilon_i)/n>0$. Thus every chore other than $p_i$ costs more than $s_i$. A bundle of at least two chores contains a chore other than $p_i$; deleting a different chore leaves residual cost above $s_i$. Hence every outer agent receives at most one chore.

The central agent receives at least two chores, since there are $n+1$ chores and only $n-1$ outer agents. She cannot receive $a$ with another chore because $A>s_n$. Also, $A<2/(n+2)$ is equivalent to $A<2(1-A)/n=2r$, so $2r>A>s_n$. Hence she cannot receive three or more $p$-chores. Thus the central agent receives exactly two $p$-chores, each outer agent receives a singleton, and one outer agent receives $a$. Every allocation of this form is weighted PROPX: the outer bundles are singletons, while deleting either chore from the central pair leaves cost $r<1/n<s_n$.

Fix such an allocation $\bX$, and let $k<n$ receive $a$. Let $\bY$ denote the allocation with $Y_n=\{a\}$, $Y_k=\{p_k,p_n\}$, and $Y_i=\{p_i\}$ for every $i<n$ with $i\neq k$. Agent $k$ is indifferent, since $c_k(Y_k)=\varepsilon_k+t_k=c_k(a)=c_k(X_k)$. Every other outer agent weakly improves: her singleton under $\bX$ costs either $\varepsilon_i$ or $t_i$, while $c_i(Y_i)=\varepsilon_i$, and $t_i-\varepsilon_i=(1-(n+2)\varepsilon_i)/n>0$. The central agent strictly improves because $c_n(Y_n)=A<2r=c_n(X_n)$. Thus $\bY$ Pareto dominates $\bX$, and no weighted PROPX allocation is PO.
\end{proof}

\begin{corollary}
\label{cor:near-equal-shares}
Let $n\geq3$ and $0<\delta<(n-2)/(n(n-1)(n+2))$. There is an $n$-agent, $(n+1)$-chore instance with strictly positive additive costs and shares $s_i=1/n-\delta$ for $i<n$ and $s_n=1/n+(n-1)\delta$ that admits no allocation that is both weighted PROPX and PO. Thus incompatibility occurs arbitrarily close to equal shares for every $n\geq3$.
\end{corollary}

\begin{proof}
The shares sum to one. Moreover, the upper bound on $\delta$ is smaller than $1/n$, so every share is positive. The first $n-1$ shares are below $1/n$, and the last is above $1/n$. The inequality $s_n<2/(n+2)$ is equivalent to the stated upper bound on $\delta$, so Theorem~\ref{thm:share-cone} applies. Since $\delta$ may be arbitrarily small, the shares may be arbitrarily close to $1/n$.
\end{proof}


Theorem~\ref{thm:share-cone} covers a region of unequal share vectors, but not every such vector. Its construction requires $n\geq3$; Theorem~\ref{thm:two-agent-weighted-counterexample} treats two agents separately. For every $n\geq3$, Corollary~\ref{cor:near-equal-shares} gives counterexamples arbitrarily close to equal shares. In this construction, $\varepsilon_i<(1-ns_i)/2=n\delta/2$, so $\varepsilon_i$ tends to zero with $\delta$. The construction therefore cannot be specialized to equal shares while retaining strictly positive costs. Whether PROPX and PO are compatible under equal shares and strictly positive additive costs remains open for $n\geq3$ when $m>n$.

\subsection{Number of chores}

For every $n\geq3$, the $(n+1)$-chore counterexamples in Theorem~\ref{thm:share-cone}, including the near-equal instances in Corollary~\ref{cor:near-equal-shares}, use the minimum possible number of chores.

\begin{proposition}
\label{prop:m-at-most-n}
If all costs and shares are strictly positive and $m\leq n$, then a weighted PROPX and PO allocation exists.
\end{proposition}

\begin{proof}
Call an allocation injective if every agent receives at most one chore. Since $m\leq n$, injective allocations exist, and each is weighted PROPX. Let $\bX$ be Pareto optimal within the finite set of injective allocations.

Suppose an allocation $\bY$ Pareto dominates $\bX$. Let $I=\{i:X_i\neq\emptyset\}$ and denote the chore in $X_i$ by $e_i$ for every $i\in I$. Every agent outside $I$ has zero cost under $\bX$ and must remain empty under $\bY$. Thus $\bY$ assigns all chores among the agents in $I$.

On vertex set $I$, draw an arc $i\to j$ if $e_i\in Y_j$. Every vertex has outdegree one, so each component contains a directed cycle. Obtain $\bZ$ by rotating the chores along these cycles; all other agents keep their bundles from $\bX$. The resulting allocation is complete and injective. If $i\to j$ lies on a cycle, then $c_j(Z_j)=c_j(e_i)\leq c_j(Y_j)\leq c_j(e_j)=c_j(X_j)$. Hence $\bZ$ weakly improves every agent.

If every vertex lies on a cycle, then $\bZ=\bY$, so some agent improves strictly. Otherwise, let $h\to j$ be an arc with $h$ outside a cycle and $j$ on that cycle, and let $p$ be the predecessor of $j$ on the cycle. Agent $j$ receives both $e_h$ and $e_p$ under $\bY$. Strict positivity gives $c_j(Z_j)=c_j(e_p)<c_j(Y_j)\leq c_j(X_j)$. Thus $\bZ$ is an injective allocation that Pareto dominates $\bX$, a contradiction. Therefore $\bX$ is weighted PROPX and PO.
\end{proof}

For two agents, compatibility extends one chore beyond the general bound $m\leq n$.

\begin{proposition}
	\label{prop:two-agents-three-chores}
	For two agents with strictly positive additive costs and arbitrary positive shares, a weighted PROPX and PO allocation exists when $m\leq3$.
\end{proposition}

\begin{proof}
The case $m\leq2$ follows from Proposition~\ref{prop:m-at-most-n}; hence let $M=\{a,b,c\}$. For each $e\in M$, let $A_e=(\{e\},M\setminus\{e\})$ and $B_e=(M\setminus\{e\},\{e\})$.
	
{An allocation that assigns all three chores to one agent is PO: any allocation that Pareto dominates it must leave the agent with the empty bundle empty and therefore coincide with the original allocation.} If either concentrated allocation is weighted PROPX, the claim follows. Otherwise, the existence of weighted PROPX allocations~\cite{DBLP:conf/www/0037L022} implies that at least one of $A_a,A_b,A_c,B_a,B_b,B_c$ is weighted PROPX.
	
{
Suppose some $A_e$ is weighted PROPX, and choose $A_e$ to minimize $c_1(e)$ among all weighted PROPX allocations of type $A$. If $A_e$ is PO, the claim follows. Otherwise, starting from $A_e$, repeatedly replace the current allocation by one that Pareto dominates it. Since the set of allocations is finite, this process terminates at a PO allocation $\bZ$. By transitivity, $\bZ$ Pareto dominates $A_e$.

Neither concentrated allocation Pareto dominates $A_e$: assigning all chores to agent~$1$ strictly raises agent~$1$'s cost, while assigning all chores to agent~$2$ strictly raises agent~$2$'s cost. Hence $\bZ=A_f$ or $\bZ=B_f$ for some $f\in M$.

First suppose $\bZ=A_f$. Necessarily $f\neq e$, since an allocation does not Pareto dominate itself. Pareto domination gives $c_1(f)\leq c_1(e)$ and $c_2(M\setminus\{f\})\leq c_2(M\setminus\{e\})$. By additivity, the latter inequality is equivalent to $c_2(e)\leq c_2(f)$. Let $\ell$ denote the third chore. Since agent~$2$ receives $\{f,\ell\}$ under $A_e$, weighted PROPX gives $\max\{c_2(f),c_2(\ell)\}\leq s_2c_2(M)$. Together with $c_2(e)\leq c_2(f)$, this shows that both chores in agent~$2$'s bundle $\{e,\ell\}$ under $A_f$ cost at most $s_2c_2(M)$. Thus $A_f=\bZ$ is weighted PROPX and PO.

It remains to suppose $\bZ=B_f$. If $f\neq e$, agent~$1$ receives $e$ and another positive-cost chore, so $c_1(M\setminus\{f\})>c_1(e)$, contradicting the fact that $B_f$ Pareto dominates $A_e$. Therefore $f=e$ and $\bZ=B_e$.
Suppose, for contradiction, that $B_e$ is not weighted PROPX. Agent~$2$ receives a singleton, so the violation is by agent~$1$. Let $M\setminus\{e\}=\{h,\ell\}$, where $c_1(h)>s_1c_1(M)$. Since $B_e$ Pareto dominates $A_e$, $c_1(M\setminus\{e\})\leq c_1(e)$ and $c_2(e)\leq c_2(M\setminus\{e\})$.
We claim that $s_2\geq1/2$. Otherwise $s_1>1/2$, and $c_1(h)>s_1c_1(M)>c_1(M)/2$. But $c_1(M\setminus\{e\})\leq c_1(e)$ implies $c_1(M\setminus\{e\})\leq c_1(M)/2$, while strict positivity gives $c_1(h)<c_1(M\setminus\{e\})$, a contradiction.
It follows that $c_2(e)\leq c_2(M)/2\leq s_2c_2(M)$. Since $A_e$ is weighted PROPX and agent~$2$ receives $\{h,\ell\}$ under $A_e$, we also have $c_2(\ell)\leq s_2c_2(M)$. Let $A_h=(\{h\},\{e,\ell\})$. Agent~$1$ receives a singleton, and both chores assigned to agent~$2$ cost at most $s_2c_2(M)$. Hence $A_h$ is weighted PROPX. Moreover,
\[
c_1(h)<c_1(h)+c_1(\ell)=c_1(M\setminus\{e\})\leq c_1(e),
\]
contradicting the choice of $A_e$. Consequently, $B_e=\bZ$ is weighted PROPX and PO.
}

If no allocation of type $A$ is weighted PROPX, some allocation of type $B$ is weighted PROPX. Interchanging the agents gives the same argument.
\end{proof}

Thus the counterexample in Theorem~\ref{thm:two-agent-weighted-counterexample} uses the minimum possible number of chores for two agents.

\section{Conclusion}

{
We consider the compatibility between weighted PROPX and PO in the allocation of indivisible chores under strictly positive additive chore costs. 
For every $n\geq2$, there is an $n$-agent weighted instance in which the two requirements are incompatible. For every $n\geq3$, incompatibility occurs with $n+1$ chores and shares arbitrarily close to equal. These bounds are sharp in the number of chores: for every $n\geq3$, no counterexample exists with at most $n$ chores, while the two-agent counterexample requires four chores. The equal-share case remains open for $n\geq3$ and $m>n$.}

\section*{Generative AI disclosure}

During the development of this work, the authors interactively used
OpenAI's ChatGPT (GPT-5.6, ``Sol'').  The central counter examples
was discovered by the system, and
they were further refined and simplified. The human authors'
contribution included asking the relevant questions, refining the arguments, imposing
coherence, and working on the writing.  All mathematical statements and proofs 
were verified by the human authors, who take full
responsibility for their correctness, originality, and presentation.







\bibliographystyle{abbrv}


%
%
%
\end{document}